\documentclass{llncs}

\usepackage{amsmath}
\usepackage{amsfonts}
\usepackage{amssymb}

\usepackage{enumerate}

\spnewtheorem*{theorem*}{Theorem}{\bfseries}{\rmfamily}
\spnewtheorem{claim}{Claim}{\itshape}{}

\newcommand{\SimS}{\mathsf{Sim}}
\newcommand{\Sim}[1]{\mathsf{Sim}\left(#1\right)}

\newcommand{\D}{\textsc{D}}
\newcommand{\F}{\textsc{F}}

\newcommand{\bigO}[1]{O\left(#1\right)}

\DeclareMathOperator*{\E}{\mathbb{ E }}

\usepackage{caption} 
\usepackage{tikz}
\usetikzlibrary{shapes,arrows,positioning, decorations.pathmorphing}

\usepackage{hyperref}
\usepackage[nameinlink]{cleveref}
\hypersetup{
  colorlinks=true,
  linkcolor=blue,
  citecolor=green
}
\crefname{table}{table}{tables}
\Crefname{table}{Table}{Tables}
\crefname{figure}{figure}{figures}
\Crefname{figure}{Figure}{Figures}
\crefname{section}{section}{sections}
\Crefname{section}{Section}{Sections}
\crefname{claim}{claim}{claims}
\Crefname{claim}{Claim}{Claims}

\title{A Time-Success Ratio Analysis of wPRF-based Leakage-Resilient Stream Ciphers}

\author{Maciej Sk\'{o}rski 
\institute{
\email{maciej.skorski@mimuw.edu.pl} \\ Cryptology and Data Security Group, University of Warsaw}
}
\begin{document}

\maketitle

\begin{abstract}
Weak pseudorandom functions (wPRFs) found an important application as main building blocks for leakage-resilient ciphers (EUROCRYPT'09). Several security bounds, based on different techniques, were given to these stream ciphers. The security loss in these reduction-based proofs is always polynomial, but has not been studied in detail. The aim of this paper is twofold. First, we present a clear comparison of quantitatively different security bounds in the literature. Second, we revisit the current proof techniques and answer the natural question of how far we are from meaningful and provable security guarantees, when instantiating weak PRFs with standard primitives (block ciphers or hash functions). In particular, we demonstrate a flaw in the recent (TCC'14) analysis of the EUROCRYPT'09 stream cipher
 Our approach is a \emph{time-to-success ratio} analysis, a universal measure introduced by Luby, which allow us to compare different security bounds.
\end{abstract}

\smallskip
\noindent \textbf{Keywords.} leakage-resilient cryptography, stream ciphers, side information, convex approximation

\section{Introduction}

\subsection{Leakage-resilient cryptography}

\paragraph{Leakage Resilience.} Traditional security notions in cryptography consider adversaries who can interact with a primitive only in a black-box manner, observing its input/output behavior. Unfortunately, this assumption is unrealistic in practice. In fact, information might leak from cryptograms at the \emph{physical implementation} layer. The attacks that capture information this way are called \emph{side-channel attacks}, and include power consumption analysis~\cite{Kocher1999}, timing attacks~\cite{Kocher1996}, fault injection attacks~\cite{Barenghi12} or memory attacks~\cite{Halderman08}.
Searching for countermeasures against side-channel attacks, one can try to prevent them modifying software or futher secure hardware. However, these techniques are more ad-hoc than generic. A completely different viewpoint is to provide primitives which are \emph{provably secure against leakage}. The research field following this paradigm is called \emph{leakage-resilient cryptography}, and has become very popular in recent years. A lot of work and progress has been done in this topic so far, since the breakthrough paper on resilient stream ciphers~\cite{Dziembowski2008}, much more than we could mention here. We refer the reader to~\cite{Alwen2010} and~\cite{Mol2010} for good surveys.

\paragraph{Modeling leakage.} A number of ways to capture the leakage has been proposed. Very first works focused on strongly restricting the type of leakage. Here we very briefly discuss most important ones, referring interested readers to surveys.
\begin{itemize}
\item{\emph{exposure resilient cryptography.}} In this line of work, the type of leakage is restricted so that adversaries learn subsets of the bits of the secret state key~\cite{Canetti2000,Dodis2001}.
\item{\emph{continuous bounded computational leakage.}} Perhaps the most pupular line of research restricting the leakage type, based on the ``only computation leaks information'' axiom introduced Micali and Reyzin~\cite{Reyzin2004}. In this modelling approach the overal execution of a cryptographic protocol is divided into time frames, and in every round leakage comes only from the parts of the internal state which are touched by computations. The amount of leakage is bounded in every round but unbounded overall. This model successfully captures side-channels attacks resulting from computation~\cite{Mol2010}, however memory attacks are more problematic as they are possible even if no computation is performed~\cite{Halderman08}. Nonetheless, leakage-resilient constructions under the ``only computation leaks information'' assumption, are of big interests~\cite{Dziembowski2008,Pietrzak09aleakage-resilient,DodisPietrzak2010,Faust2012,Yu2013}, to mention only stream-ciphers related works. We also note that in specific cases, in particular for stream ciphers we will be interested in, the authors argue that their security models go beyond the ``only computation leaks information'' assumption and actually capture memory attacks (cf.~\cite{Pietrzak09aleakage-resilient}).
\item{\emph{probing attacks}}. In this approach, initiated in~\cite{ISW03}, adversaries can learn or influence the values at some wires, during the evaluation of a cicuit.
\item{\emph{auxiliary inputs.}} The works~\cite{DodisKalai2009,DodisGoldwasser2010} study a setting where adversaries can learn a function of the secret state, which is hard to invert. It allows leaking information larger than the size of the secret state and is believed to be most practical. However, it is also considered very challenging for proving security of constructions.
\end{itemize}
Being interested in leakage-resilient stream ciphers, we follow the related works and focus on continuous computational leakage  through this paper (see \Cref{sec:prelim} for a formal definition in the concrete setting).

\subsection{Leakage-resilient stream ciphers.}
\paragraph{What are stream ciphers?} The purpose of stream ciphers is to efficiently encrypt data streams of arbitrary length. The most popular constructio mimics the one-time pad encryption, by deploying a  generator which stretches the initial randomness into a keystream. Such a generator, when initialized with a secret state, recursively computes a sequence of output blocks where the security requirement is that the last part look random given the previous outputs. 

\paragraph{Leakage-resilient design.} The main concern in proving leakage-resilience is that the keystream generator must be secure against leakages, which appear in every round (in the continuous leakage model). Such a generator could be deployed either with a pseudorandom generator and extractor~\cite{Dziembowski2008}, or a weak pseudorandom functions~\cite{Pietrzak09aleakage-resilient,YuStandaert2010,Faust2012,Yu2013}. In any case, the idea is to refresh the secret state (key) in every round, to make compromising it possibly difficult. Below we briefly discuss some advantages of the second approach, and return to a more detailed discussion of the concrete designs in \Cref{sec:design}.

\paragraph{Why wPRFs-based design?} Informally, pseudorandom functions look random on many adversarialy chosen inputs (under a uniform secret key), whereas weak pseudorandom functions look random only on random inputs. Below we elaborate more on why weak pseudorandom functions are of special interests for leakage-resilient stream cipher constructions.
\begin{enumerate}[(a)]
\item From a high-level viewpoint, we have at least two very good reasons to build leakage-resilient stream ciphers using weak PRFs, as proposed in~\cite{Pietrzak09aleakage-resilient}. 
First, this approach is simple and thus more efficient to implement and much easier to analyze than the original proposal~\cite{Dziembowski2008}, which combines a pseudorandom generator and an extractor. Second, and most important, it less vurnelable to side-channel attacks and more reliable from a practical viewpoint. This is because the construction can be instantiated with only one component - a weak PRF. Mounting an attack against one component is less likely, as opposed to the original construcion~\cite{Standaert2010,Standaert2011,Standaert2012}. Moreover, this construction is more reliable from a practical viewpoint when instantiated with block ciphers understood as weak PRFs (like AES), because their security against side-channel attacks has been carefully analyzed.
\item
From a technical viewpoint, weak pseudorandom functions are primitives very pleasurable to deal with in the context of leakage. As opposed to (strong) pseudorandom functions they can be shown to remain secure with weak keys (that is when keys are not uniform but have some entropy defficiency), which is the key ingredient of the cipher resilience proof. Security with weak keys can be proven either by a computational variant of the Dense Model Theorem~\cite{Pietrzak09aleakage-resilient} or by a recent techniques involving the square-security notion~\cite{Dodis2013}.
\end{enumerate}

\paragraph{Security.}
The stream cipher is considered secure, if for a sequence of its outputs, the last round output block looks pseudorandom, given the outputs from previous rounds. See \Cref{sec:prelim} for a formal definition.

\subsection{Reductions Quality Issues.}
The security of leakage-resilient stream ciphers is always proven by a reduction to underlaying more standard components, as pseudorandom generators, extractos, pseudorandom functions, whose security is generally well understood. Proving these bounds is challenging and still we can only prove quite poor bounds, unless we impose strong idealistic assumptions. Below we elaborate more on this topic.
\begin{enumerate}[(a)]
\item \emph{Significant security losses in the standard model}. 
Reduction proofs yield quite weak bounds, and this is common for all related works. For leakage-resilient stream ciphers we have to lose a constant fraction of the security compared to its original level, even if the leakage is just one bit!
\item \emph{No provable security with standard building bricks.} When we aim for the (provable) security level recommended nowadays, which is at least $80$ bits, we need to start with primitives (like block ciphers) whose security is bigger than $400$ bits, given current knowledge. This is a direct consequence of the issue with weak reductions we mentioned above. 
\item \emph{Different bounds are hard to compare.} Depending on the technique, different bounds are obtained. Formulas offer security against different adversarial profiles - running time, success probability, leakage length.
\end{enumerate}

\subsection{Problem and results, informally}

Motivated in studying the quality of reductions, we state our problem as a series of questions. We briefly answer them here, announcing our results informally, and discuss in the next section in more detail.
\begin{quote}
\textbf{Q1}: How tight are reduction-based security proofs for leakage-resilient stream ciphers?
\end{quote}
We revisit the best known bounds and analyze the tightness of reductions using time-success ratios. We discuss these tools in more detail in the next section.
\begin{quote}
\textbf{A1}: All results loses more than $75\%$ of the original security (measured in bits), paying for the resilience feature. This holds even for one bit of leakage per computation!
\end{quote}
The second issue we address is how far we really are from having provable security for constructions instantiated from practically used components.
\begin{quote}
\textbf{Q2}: Can we instantiate a leakage-resilient stream cipher, provable secure in the standard model, with a standard (128 or 256-bit) block cipher as a weak PRF?
\end{quote}
The most serious attempt to achieve meaningful security using standard 256-bit block ciphers is due to Pietrzak and Jetchev~\cite{JetchevP14}. They improved and simplified bounds for the EUROCRYPT'09 stream cipher. However, as we will explain later, the better of the two claimed bounds doesn't apply because of a flaw in the proof~\cite{Pietrzak15_private}.
\begin{quote}
\textbf{A2}: No, given the current state of art. The recent analysis from TCC'14 which gives an affirmative answer,
contains a flaw. We will discuss it in \Cref{sec:Pietrzak_error}.
\end{quote}
Because of the lack of a positive answer above, it is natural to ask how strong our starting primitive needs to be, given current proof techniques. We believe that it is of interests to know how far we are with provable secure bounds from the idealized bounds, especially that this approach seems to be relatively rarely taken.
\begin{quote}
\textbf{Q3}: What a weak PRF do we need to achieve the recommended security level of $80$ bits, given the known techniques?
\end{quote}
Using our time-success ratio analysis we given an answer
\begin{quote}
\textbf{A3}: At least with $512$ bits of security (and assuming small leakage). We propose to instantiate with $\mathsf{SHA512}$ as a weak PRF.
\end{quote}

\subsection{Results and techniques in details.} 

\paragraph{Flaws in the recent analysis of the EUROCRYPT'09 stream cipher.} Pietrzak and Jetchev came up with an elegant idea to simplify the security proof of the EUROCRYPT'09 stream cipher built from a weak PRF. To this end, they prove a theorem about simulating auxiliars inputs, giving two alternative proofs~\cite{JetchevP14}.
One of them would imply good security in the standard model, with $\mathsf{AES}$ used as the weak PRF (for the first time). Unfortunately, as we point out in~\Cref{sec:Pietrzak_error} in this paper, the proof of this stronger bound is wrong. For this reason, only the second much weaker bound applies so we cannot prove meaningful security instantiating the stream cipher with a standard 256 block cipher, like $\mathsf{AES}$. 

\paragraph{An improved simulator for auxiliary inputs and better security for the EUROCRYPT'09 stream cipher.} We don't know how to fix the issue with the flawed analysis in~\cite{JetchevP14}. However, we improve the alternative proof of the simulating lemma by a significant factor, which gives a better analysis of the stream cipher than the~\cite{JetchevP14}. Our proof might of independent interest because of the proof technique, which utilizes a variant of the Baron-Maurey approximation theorem. We refer the reader to \Cref{thm:simulator_new} in \Cref{sec:Pietrzak_error} for more details.

\paragraph{A framework to compare different reductions.} Bellare,Rogaway~\cite{Bellare1996} were first who emphasized the importance of studying the tightness of security proofs in practical applications. Following the approach proposed by Luby~\cite{Luby1994}, based on time-success ratio (see~\Cref{sec:time-success}), we provide a general tool for determining the security of every stream cipher reducing to a weak PRF. Technically, by constrained optimization we determine the time-success ratio of a stream cipher from the security of its main building component. This approach is used in different area of provable security (cf.~\cite{Buldas2013} and many similar works), but to our knowledge has never been taken with respect to leakage-resilient stram ciphers (in particular in all the works we cite). 

\paragraph{A clear security loss formula.} We abstract the ``typical form'' for the loss in most reductions from a stream cipher to the underlying weak PRF. Namely the time/advantage pairs, describing adversarial resources and success probability, for the original primitive $(s,\epsilon)$ and for the cipher $(s',\epsilon')$ are related as $\epsilon'=\epsilon^{A}$ and $s' = s\cdot \epsilon^{C}-\epsilon^{-B}$ for some \emph{explicit constants $A,B,C$} in the exponents. Extending slightly this model to capture leakage-depended factors, we actually cover all related works. We solve the related optimization program and show how \emph{explicitly} the time-success ratio degradation depends on these constants (see ~\Cref{sec:time-success_simpler}). It turns out that remains is a fraction of roughly $\sim \frac{A}{B+C+1}$ of the original security (measured in bits). For all known constructions, this is smaller than $25\%$. 

\paragraph{A survey of known results.} We present the time-success ratio analysis of wPRF-based leakge-resilient stream ciphers. The lack of such results is perhaps partially because of complicated formulas, and partially because in folklore these bounds are considered mainly of theoretical interests. Yet, we believe that comparing these bounds is interesting, in particular with respect to the ``dream bounds'' corresponding to the flawed analysis in~\cite{JetchevP14}, which - if can be proven - gives a much better security level than other techniques. For more details, see~\Cref{sec:time-success_sc_survey}.

\section{Preliminaries}\label{sec:prelim}

\subsection{Leakage resilient cryptography}\label{sec:leakage_resilient}

We start with the definition of weak pseudorandom functions, which are \emph{computationaly indistinguishable} from random functions, when queried on random inputs and fed with iniform secret key.
\begin{definition}[Weak pseudorandom functions]
A function $\F: \{0, 1\}^{k} \times \{0, 1\}^{n} \rightarrow \{0, 1\}^{m}$ is an $(\epsilon, s, q)$-secure weak PRF if its outputs on $q$ random inputs are indistinguishable from random by any distinguisher of size $s$, that is 
\begin{align*}
\left| \Pr \left[\D\left(\left( X_i \right)_{i=1}^{q},\F((K,X_i)_{i=1}^{q} \right)=1\right] - \Pr \left[\D\left(\left(X_i\right)_{i=1}^{q},\left(R_i\right)_{i=1}^{q}\right)=1 \right] \right| \leqslant \epsilon
\end{align*}
where the probability is over the choice of the random $X_i \leftarrow \{0,1\}^n$, the choice of a random key $K \leftarrow \{0,1\}^k$ and $R_i \leftarrow \{ 0,1\}^m$ conditioned on $R_i = R_j$ if $X_i = X_j$ for some $j < i$.
\end{definition}
Stream ciphers generate a keystream in a recursive manner. The security requires the output stream should be indistinguishable from uniform\footnote{We note that in a more standard notion the entire stream $X_1,\ldots,X_{q}$ is indistinguishable from random. This is implied by the notion above by a standard hybrid argument, with a loss of a multiplicative factor of $q$ in the distinguishing advantage.}.
\begin{definition}[Stream ciphers]
A \emph{stream-cipher} $\mathsf{SC} : \{0, 1\}^k \rightarrow \{0, 1\}^k \times \{0, 1\}^n$ is a function that need to be initialized with a secret state $S_0 \in \{0, 1\}^k$ and produces a sequence of output blocks $X_1, X_2, . . . $ computed as
\begin{align*}
 (S_i, X_i) := \mathsf{SC}(S_{i-1}).
\end{align*}
A stream cipher  $\mathsf{SC}$ is $(\epsilon,s,q)$-secure if for all $1 \leqslant i \leqslant q$, the random variable $X_i$ is $(s,\epsilon)$-pseudorandom given $X_1, . . . , X_{i-1}$ (the probability is also over the choice of the initial random key $S_0$).
\end{definition}
Now we define the security of leakage resilient stream ciphers, which follow the ``only computation leaks'' assumption.
\begin{definition}[Leakage-resilient stream ciphers]
A leakage-resilient stream-cipher is $(\epsilon,s,q,\lambda)$-secure if it is $(\epsilon,s,q)$-secure as defined above, but where the distinguisher in the $j$-th round gets $\lambda$ bits of arbitrary deceptively chosen leakage about the secret state accessed during this round. More precisely, before $(S_j,X_j) := \mathsf{SC}(S_{j−1})$ is computed, the distinguisher can choose any leakage function $f_j$ with range $\{0,1\}^{\lambda}$, and then not only get $X_j$, but also $\Lambda_j := f_j(\hat{S}_{j−1})$, where $\hat{S}_{j−1}$ denotes the part of the secret state that was modified (i.e., read and/or overwritten) in the computation $\mathsf{SC}(S_{j−1})$.
\end{definition}

\subsection{Time-Success Ratio}\label{sec:time-success}
The running time (circuit size) $s$ and success probability $\epsilon$ of attacks (practical and theoretical) aggainst a particular primitive or protocol may vary. For this reason Luby~\cite{Luby1994} introduced the time-success ratio $\frac{t}{\epsilon}$ as a universal measure of security. This model widely used to analyze security, cf. \cite{Buldas2013} and related works.
\begin{definition}[Security by Time-Success Ratio~\cite{Luby1994}]\label{def:time-success}
A primitive $P$ is said to be $2^{k}$-secure if for \emph{every adversary} with time resources (circuit size in the nonuniform model) $s$, the success probability in breaking $P$ (advantage) is at most $\epsilon < s\cdot 2^{-k}$. We also say that the time-success ratio of $P$ is $2^{k}$, or that is has $k$ bits of security.
\end{definition}
For example, $\mathsf{AES}$ with a $256$-bit random key is believed to have $256$ bits of security as a \emph{weak} PRF\footnote{We consider the security of $\mathsf{AES256}$ as a weak PRF, and not a standard PRF, because of non-uniform attacks which show that 
no PRF with a $k$ bit key can have $s/\epsilon \approx 2 ^k$ security~\cite{DeTT09}, at least unless we additionally require $\epsilon \gg 2^{-k/2}$.}.

\section{Leakage-Resilient Stream Ciphers Design}\label{sec:design} 
In this section we briefly discuss the known constructions of leakage-resilient stream ciphers in the standard model (without random-oracle assumptions)
\subsection{The very first idea (FOCS'08)}
The first construction of leakage-resilient stream cipher was proposed by Dziembowski and Pietrzak in~\cite{Dziembowski2008}. It has the characterstic \emph{alternating structure} which allows for proving security against \emph{adaptively chosen leakage}. 
\subsection{A construction based on a wPRF (EUROCRYPT'09)}
On \Cref{fig:SC_Pietrzak} below we present a simplified version of this cipher~\cite{Pietrzak09aleakage-resilient} based on a weak pseudorandom function (wPRF). A weak pseudorandom function is a primitive which ``looks'' like a random function when queried on random inputs, see~\Cref{sec:prelim} for a formal definition.
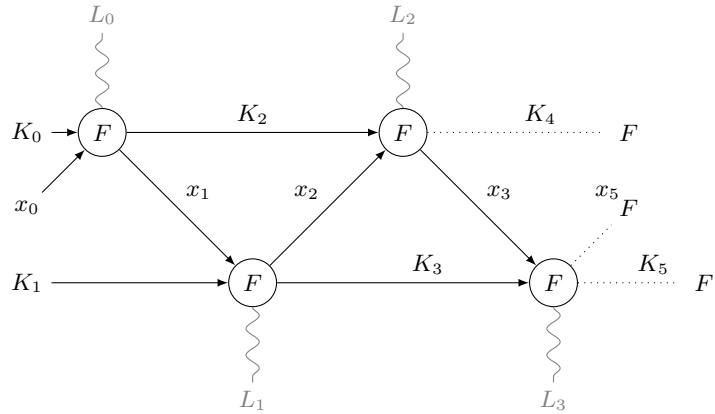
\begin{figure}[!th]
\centering
\begin{tikzpicture}
\node[] (k0) at (0,2) {$K_0$};
\node[] (x0) at (0,1) {$x_0$};
\node[] (k1) at (0,0) {$K_1$};
\node[draw,circle] (f0) at (1,2) {$F$};
\node[draw,circle] (f1) at (3,0) {$F$};
\node[draw,circle] (f2) at (5,2) {$F$};
\node[draw,circle] (f3) at (7,0) {$F$};
\node[circle, text opacity = 0] (f4) at (8,2) {$F$};
\node[circle, text opacity = 0] (f41) at (8,1) {$F$};
\node[circle, text opacity = 0] (f5) at (9,0) {$F$};
\draw[-latex] (k0) -- (f0);
\draw[-latex] (x0) -- (f0);
\draw[-latex] (k1) -- (f1);
\draw[-latex] (f0) -- (f1) node[midway, above right] {$x_1$};
\draw[-latex] (f1) -- (f3) node[midway, above right] {$K_3$};
\draw[-latex] (f1) -- (f2) node[midway, above left] {$x_2$};
\draw[-latex] (f0) -- (f2) node[midway, above] {$K_2$};
\draw[-latex] (f2) -- (f3) node[midway, above right] {$x_3$};
\draw[dotted] (f2) -- (f4) node[midway, above right] {$K_4$};
\draw[dotted] (f3) -- (f5) node[midway, above right] {$K_5$};
\draw[dotted] (f3) -- (f41) node[above left] {$x_5$};
\node[color=gray, above = 1 of f0] (l0) {$L_0$};
\draw[color=gray, decorate, decoration=snake] (f0) -- (l0);
\node[color=gray, above= 1 of f2] (l2) {$L_2$};
\draw[color=gray, decorate, decoration=snake] (f2) -- (l2);
\node[color=gray, below= 1 of f1] (l1) {$L_1$};
\draw[color=gray, decorate, decoration=snake] (f1) -- (l1);
\node[color=gray, below= 1 of f3] (l3) {$L_3$};
\draw[color=gray, decorate, decoration=snake] (f3) -- (l3);
\end{tikzpicture}
\caption{The EUROCRYPT'09 stream cipher (adaptive leakage). $F$ denotes a weak PRF. By $K_i$ and $x_i$ we denote, respectively, values of the secret state and keystream bits. Leakages are denotted in gray with $L_i$.}
\label{fig:SC_Pietrzak}
\end{figure}
\subsection{Saving key randomness (CSS'10, CHESS'12)}
\noindent A slightly different approach is proposed in~\cite{YuStandaert2010}. The authors argue that side-channel attacks in practice are mounted against a specific target, and require speficic measurements equipment; thus adaptive security is somewhat an overkill. The second observation is that the cipher in~\cite{Pietrzak09aleakage-resilient} seems to waste lots of randomness, because the security in best case is only comparable to the length of one secret key, whereas the cipher is initialized with two random keys (denoted with $K_0,K_1$ on~\Cref{fig:SC_Pietrzak}). They remove the alternating structure and use only one key and two alternating public random values, aiming at (weaker) non-adaptive security. Unfortunately, the proof that these two alternating public values are enough were wrong, as pointed out in~\cite{Faust2012}. However one gets provable non-adaptive security, assuming that every round uses fresh randomness~\cite{Faust2012}. Such a big amount of randomness makes the cipher inpractical, but the authors show how to reduce it further. Summing up, one gets only non-adaptive security but saves secret randomness replacing the ``wasted'' key by a public string.
The scheme is illustrated in~\Cref{fig:SC_Faust} below.
\begin{figure}[]
\centering
\begin{tikzpicture}
\node[] (k0) at (1,0) {$K_0$};
\node[draw,circle] (f0) at (2,0) {$F$};
\node[draw,circle] (f1) at (5,0) {$F$};
\node[draw,circle] (f2) at (8,0) {$F$};
\node[circle] (f3) at (10,0) {};
\draw[-latex] (k0) -- (f0);
\draw[-latex] (f0) -- (f1) node[midway, above] {$K_1$};
\draw[-latex] (f1) -- (f2) node[midway, above] {$K_2$};
\draw[dotted] (f2) -- (f3) node[midway, above right] {$K_3$};
\node[above = 1 of f0] (p0) {$p_0$};
\node[above = 1 of f1] (p1) {$p_1$};
\node[above = 1 of f2] (p2) {$p_2$};
\draw[-latex] (p0) -- (f0);
\draw[-latex] (p1) -- (f1);
\draw[-latex] (p2) -- (f2);
\node[color=gray, below = 1 of f0] (l0) {$L_0$};
\node[color=gray, below = 1 of f1] (l1) {$L_1$};
\node[color=gray, below = 1 of f2] (l2) {$L_2$};
\draw[color=gray, decorate, decoration=snake] (f0)--(l0);
\draw[color=gray, decorate, decoration=snake] (f1)--(l1);
\draw[color=gray, decorate, decoration=snake] (f2)--(l2);
\node[below right = 1 of f0] (x0) {$x_0$};
\node[below right = 1 of f1] (x1) {$x_1$};
\node[below right = 1 of f2] (x2) {$x_2$};
\draw[-latex] (f0) -- (x0);
\draw[-latex] (f1) -- (x1);
\draw[-latex] (f2) -- (x2);
\end{tikzpicture}
\caption{The CSS'10/CHESS'12 stream cipher. $F$ denotes a waek PRF. By $K_i$ and $x_i$ we denote, respectively,  the values of secret state and keystream bits. Leakages are denoten in gray with $L_i$. The cipher requires public independent random values $p_i$.}
\label{fig:SC_Faust}
\end{figure}
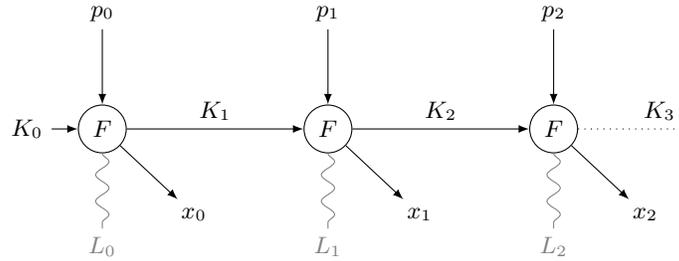

\subsection{Saving public randomness (CT-RSA'13)}\label{sec:CT-RSA13}

The problem with large public randomness, required for the last cipher, was addressed in~\cite{YuStandaert2010}. The public values, required in the previous construction, are generated on-the-fly from a single public value, by running a strong PRF in counter mode on it. For an illustration, see~\Cref{fig:SC_YuStandaert} below. 
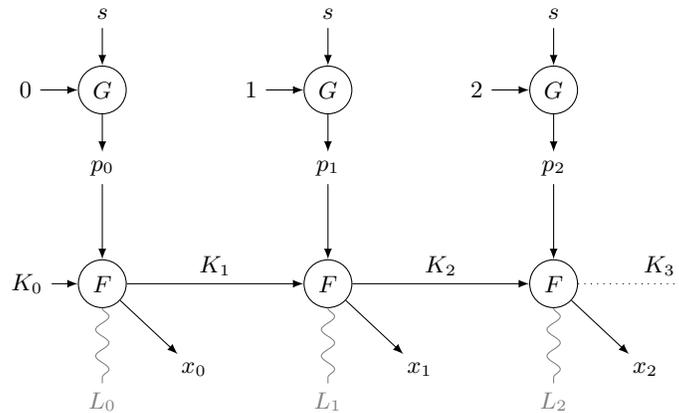
\begin{figure}[]
\centering
\begin{tikzpicture}
\node[] (k0) at (1,0) {$K_0$};
\node[draw,circle] (f0) at (2,0) {$F$};
\node[draw,circle] (f1) at (5,0) {$F$};
\node[draw,circle] (f2) at (8,0) {$F$};
\node[circle] (f3) at (10,0) {};
\draw[-latex] (k0) -- (f0);
\draw[-latex] (f0) -- (f1) node[midway, above] {$K_1$};
\draw[-latex] (f1) -- (f2) node[midway, above] {$K_2$};
\draw[dotted] (f2) -- (f3) node[midway, above right] {$K_3$};
\node[above = 1 of f0] (p0) {$p_0$};
\node[above = 1 of f1] (p1) {$p_1$};
\node[above = 1 of f2] (p2) {$p_2$};
\node[draw, circle, above = 0.5 of p0] (g0) {$G$};
\node[draw, circle, above = 0.5 of p1] (g1) {$G$};
\node[draw, circle, above = 0.5 of p2] (g2) {$G$};
\node[left = 0.5 of g0] (c0) {$0$};
\node[left = 0.5 of g1] (c1) {$1$};
\node[left = 0.5 of g2] (c2) {$2$};
\node[above = 0.5 of g0] (s0) {$s$};
\node[above = 0.5 of g1] (s1) {$s$};
\node[above = 0.5 of g2] (s2) {$s$};
\draw[-latex] (s0) -- (g0);
\draw[-latex] (s1) -- (g1);
\draw[-latex] (s2) -- (g2);
\draw[-latex] (p0) -- (f0);
\draw[-latex] (p1) -- (f1);
\draw[-latex] (p2) -- (f2);
\draw[-latex] (f0) -- (x0);
\draw[-latex] (f1) -- (x1);
\draw[-latex] (f2) -- (x2);
\draw[-latex] (g0) -- (p0);
\draw[-latex] (g1) -- (p1);
\draw[-latex] (g2) -- (p2);
\draw[-latex] (c0) -- (g0);
\draw[-latex] (c1) -- (g1);
\draw[-latex] (c2) -- (g2);

\node[below right = 1 of f0] (x0) {$x_0$};
\node[below right = 1 of f1] (x1) {$x_1$};
\node[below right = 1 of f2] (x2) {$x_2$};
\node[color=gray, below = 1 of f0] (l0) {$L_0$};
\node[color=gray, below = 1 of f1] (l1) {$L_1$};
\node[color=gray, below = 1 of f2] (l2) {$L_2$};
\draw[color=gray, decorate, decoration=snake] (f0)--(l0);
\draw[color=gray, decorate, decoration=snake] (f1)--(l1);
\draw[color=gray, decorate, decoration=snake] (f2)--(l2);
\end{tikzpicture}
\caption{The CTR-RSA'13 stream cipher (nonadaptive leakage, \textsf{minicrypt}). $F$ is a weak PRF and $G$ is a strong PRF. By $K_i$ and $x_i$ we denote, respectively,  the values of secret state and keystream bits. Leakages are denoted in gray with $L_i$. The function $F$ generating the keystream is rerandomized using values $p_i$, produced by $G$ in counter mode from the single public seed $s$.}
\label{fig:SC_YuStandaert}
\end{figure}
The result is only conditional and holds in the hypothetical world \textsf{minicrypt}, where one-way functions exist, but there is no public-key cryptography. Still, it may be a good clue on what we should aim for, when we want provable security in the standard model.

\section{Results}\label{sec:main}

\subsection{The time-success ratio under reductions}\label{sec:time-success_general}
We consider first a very abstract setting, where a primitive $P'$ is built from $P$. Assume, that the security of $P'$ reduces to the security of $P$ in the following quantitative way: 
\begin{quote}
\textbf{R}: If $P$ is secure against an adversary $(s,\epsilon)$, then $P'$ is secure against any adversary $(s',\epsilon')$, where 
\begin{align}\label{eq:reduction}
\begin{array}{rl}
s' & = p(s,\epsilon),\\  \epsilon' & = q(s,\epsilon)
\end{array}
\end{align}
for some functions $p(\cdot),q(\cdot)$.
\end{quote}
In the simpletst case, the functions $p(\cdot)$ and $q(\cdot)$ are algebraic functions of original parameters, like $\epsilon' = \epsilon^{1/2}$ or $s'=s\epsilon^{2}$ (the second case appears particularly often as a result of the Chernoff Bounds). In leakage-resilient cryptography these formulas are more complicated and typically involve some additional parameters, like the leakage length or the number of queries. The natural question here is how the security, understood as in~\Cref{def:time-success}, of the two primitives $P$ and $P'$ are related to each other. 
Before we give the answer (the proof appears in \Cref{proof:thm:time-success_general}.
\begin{theorem}[The time-success ratio as min-max optimization]\label{thm:time-success_general}
Let security of $P'$ reduces to security of $P$ as in~\Cref{eq:reduction}. If $P$ has $k$ bits of security then $P'$ has $k'$ bits of security where $k'$ is the maximal value such that the following program
\begin{align}\label{eq:optimization_0}
\begin{aligned}
 \underset{(s',\epsilon')}{\mathrm{minimize}} \ \underset{(s,\epsilon)}{\mathrm{maximize}} && & \frac{s'}{\epsilon'} & \\
\textrm{s.t.} 
&& & \frac{s'}{\epsilon'} \leqslant 2^{k'}, 1\leqslant s',\ 0\leqslant \epsilon' \\
&& & \frac{s}{\epsilon} \leqslant 2^{k}, \ 1\leqslant s,\ 0\leqslant \epsilon \\
&& & s'  \leqslant p(s,\epsilon), \ \epsilon'  \geqslant q(s,\epsilon) \\
\end{aligned}
\end{align}
 has a \underline{positive finite value}. 
\end{theorem}
\begin{remark}
If we cannot find a pair $(s,\epsilon)$ corresponding to $(s',\epsilon')$ then the feasible set in~\Cref{eq:optimization_0} is empty, so that the value of the program becomes $-\infty$.
\end{remark}

\subsection{The time-success ratio under algebraic transformations}\label{sec:time-success_simpler}

In the most typical case we can solve \Cref{eq:optimization_0} explicitly, as shown by \Cref{thm:reduction_simpler} below.
\begin{theorem}[Time-success ratio for algebraic transformations]\label{thm:reduction_simpler}
Let $a,b,c$ and $A,B,C$ be positive constants. Suppose that $P'$ is secure against adversaries $(s',\epsilon')$, whenever $P$ is secure against adversaries $(s,\epsilon)$, where
\begin{align}\label{eq:reduction_simpler}
\begin{array}{rl}
 s' & = s\cdot c\epsilon^{C} - b\epsilon^{-B} \\ 
 \epsilon' & = a\epsilon^A.
\end{array}
\end{align}
In addition, suppose that the following condition is satisfied
\begin{align}\label{eq:small_exponent}
 A \leqslant C+1.
\end{align}
Then the following is true: if $P$ is $2^{k}$-secure, then $P'$ is $2^{k'}$-secure where
\begin{align}
k' = \left\{
\begin{array}{rl}
\frac{A}{B+C+1} k + \frac{A}{B+C+1}(\log c - \log b)-\log a, & \quad b\geqslant 1 \\
\frac{A}{C+1} k + \frac{A}{C+1}\log c -\log a, & \quad b = 0
\end{array}
\right.
\end{align}
\end{theorem}
The proof is elementary though not immediate. It appears in \Cref{proof:thm:reduction_simpler}.
\begin{remark}[On the technical condition \eqref{eq:small_exponent}]
This condition is satisfied in almost all applications, at in the reduction proof typically $\epsilon'$ cannot be better (meaning higher exponent) than $\epsilon$. Thus, quite often we have $A\leqslant 1$.
\end{remark}

\subsection{An error in the recent EUROCRYPT'09 stream cipher analysis (TCC'13) and our improvement}\label{sec:Pietrzak_error}

\paragraph{Simulating auxiliary inputs.}
In~\cite{JetchevP14} there is the following theorem (here we state the correct version~\cite{Pietrzak15_private}): 
\begin{lemma}[Simulating auxiliary inputs]
For any random variable $X\in\{0,1\}^n$, any correlated $Z\in\{0,1\}^\lambda$ and every choice of parameters $(\epsilon,s)$ there is a randomized function $\SimS:\{0,1\}^n\rightarrow \{0,1\}^\lambda$ of complexity $\bigO{s\cdot 2^{4\lambda}\epsilon^{-4}}$ such that $Z$ and $\Sim{X}$ are $(\epsilon,s)$-indistinguishable given $X$.
\end{lemma}
This theorem is the core of the improved analysis of the EUROCRYPT'09 stream cipher. Using it, as decribed in~\cite{JetchevP14}, one proves the resilience of the cipher if the underlying weak PRF is $(s,\epsilon)$-secure against two queries on random inputs.
\paragraph{More on the flaws.}
In the claimed better bound $\bigO{s\cdot 2^{3\lambda}\epsilon^{-2}}$ there is a mistake on page 18 (eprint version), when the authors enforce a signed measure to be a probability measure by a mass shifting argument. The number $M$ defined there is in fact a function of $x$ and is hard to compute. The original proof asumes that this is a constant independent of $x$. In the alternative bound $\bigO{s\cdot 2^{3\lambda}\epsilon^{-2}}$ a fixable flaw is a missing factor of $2^{\lambda}$ in the complexity (page 16 in the eprint version), which is because what is constructed in the proof is only a probability mass function, not yet a sampler~\cite{Pietrzak15_private}.

\paragraph{Our improvement.} We don't know how to reduce the exponent in $\epsilon$. However, we can improve the constant in the exponent of $\lambda$, from $4$ to $2$. This is significant for the application to the cipher, as we improve its security by a factor of $2^{\Theta(\lambda)}$, which is typically of order $\Theta(\epsilon^{-1})$ (see~\cite{JetchevP14,Pietrzak09aleakage-resilient}).

\begin{theorem}[Better simulating auxiliary inputs]\label{thm:simulator_new}
for every distribution $X,Z$ on $\mathcal{X}\times\{0,1\}^{\lambda}$ and every $\epsilon$, $s$ there exists a ``simulator'' $h:\mathcal{X}\rightarrow \{0,1\}^\lambda$ such that (a) the distributions of 
$(X,\mathsf{h}(X))$  and $(X,Z)$ are $(s,\epsilon)$-indistinguishable and (b) $\mathsf{h}$ is of complexity $s_{\mathsf{h}} = \bigO{  s\cdot 2^{2\lambda}\epsilon^{-4}}$.
\end{theorem}

\section{Survey of security bounds}\label{sec:time-success_sc_survey}

In \Cref{table:1} below we present the comparison of different bounds for leakage-resilient stream ciphers built from weak PRFs. We assume that the number of blocks $q$ is constant. Wihtout loosing generality, we assume that the time-advantage ratio for our PRF is constant, that is $s/\epsilon \approx 2^k$ where $k$ is the key length. This corresponds to the assumption that the best attack is a brute-force search\footnote{This is not the case of assymetric primitives: consider e.g.RSA,here given our current understanding of the hardness of factoring, $\epsilon$ goes from basically 0 to 1 as the running time $s$ reaches the time required to run the best factoring algorithms.}. This assumption is reasonable, for example best block ciphers like $\mathsf{AES}$ are believed to have such security as PRFs. The security level is computed from \Cref{thm:reduction_simpler} by putting the bounds from the related works (we omit computations).

\begin{table}
\centering
\resizebox{0.95\textwidth}{!}{
\begin{tabular}{|c|l|l|l|c|}
\hline
 Cipher & Analysis & Proof techniques & Security level & Comments  \\
\hline
(1) & \cite{Pietrzak09aleakage-resilient} & Pseudoentropy chain rules & $k' \ll \frac{1}{8}k$& large number of blocks \\
\hline
(1) & \cite{JetchevP14} & Aux. Inputs Simulator (corr.) & $k' \approx \frac{k}{6}-\frac{5}{6}\lambda $ & \\
\hline
(1) & \cite{VadhanZ13} & Aux. Inputs Simulator & $k' \approx \frac{k}{6}-\frac{1}{3}\lambda $ & \\
\hline
(1) & \textbf{This work} & Aux. Inputs Simulator (impr.) & $k' \approx \frac{k}{6}-\frac{1}{2}\lambda $ & \\
\hline
(1) & \emph{Dream bound} & Aux. Inputs Simulator (impr.) & $k' \approx \frac{k}{4}-\lambda $ & unproven (the flaw)\\
\hline
(2) & \cite{Faust2012} & Pseudoentropy chain rules & $k' \approx \frac{k}{5} - \frac{3}{5}\lambda$ & large public seed \\
\hline
(3) & \cite{Yu2013} & Square-friendly apps. &  $k' \approx \frac{k}{4}-\frac{3}{4}\lambda$ & only in $\mathsf{minicrypt}$ \\ \hline
\end{tabular}
}

\caption{Different bounds for wPRF-based leakage-resilient stream ciphers. $k$ is the length of the secret key for the wPRF. The value $k'$ is the security level for the cipher, understood in terms of time-success ratio. the numbers denote: (1) The EUROCRYPT'09 cipher, (2) The CSS'10/CHESS'12 cipher, (3) The CT-RSA'13 cipher.}
\label{table:1}
\end{table}
\noindent It seems that the best cipher (in the standard model) is the EUROCRYPT'09 cipher. It provides the adaptive security in the standard model and loses about $\frac{5}{6}$ of its original security (the best analysis is due to Vadhan, the second best is this paper). The CSS'10/CHESS'12 loses about $\frac{4}{5}$ of its original security but requires large public randomness.

\bibliographystyle{amsalpha}
\bibliography{./citations}

\appendix

\section{Proof of \Cref{thm:time-success_general}}\label{proof:thm:time-success_general}

We notice that we are looking for the biggest value $k'$ such that for \emph{every} $(s',\epsilon')$ satysfying $s'\geqslant 1,\ \epsilon'\geqslant 0,\ 2^{k'}\geqslant s'/\epsilon'$ there exist \emph{some} values $(s,\epsilon)$ such that $s\geqslant 1$, $\epsilon \geqslant 0$, $s'\leqslant p(s,\epsilon)$, $\epsilon'\geqslant q(s,\epsilon)$ and $2^{k}\geqslant s/\epsilon $.
 for given values $(s',\epsilon')$ we can choose $(s,\epsilon)$ so that the ratio $s'/\epsilon'$ is possibly maximal, provided that the constraint $s/\epsilon \leqslant 2^k$ is satisfied. 
Taking into account the quantifiers \emph{every} and \emph{some} we get the following min-max characterization

\section{Proof of \Cref{thm:reduction_simpler}}\label{proof:thm:reduction_simpler}

\begin{proof}
Consider the program in \Cref{thm:time-success_general}. In our setting we have
\begin{align*}
\begin{array}{rl}
 p(s,\epsilon) &= s\cdot c\epsilon^{C} - b\epsilon^{-B} \\
 q(s,\epsilon) & =  a\epsilon^A
\end{array}
\end{align*}
The constraint $s' \geqslant 1$ is equivalent to 
\begin{align}\label{eq:optimization_constraint_1}
s \geqslant c^{-1}(1+b\epsilon^{-B})\epsilon^{-C}
\end{align}
Thus, the all constraints on $s$ can be written as
\begin{align*}
c^{-1}(1+b\epsilon^{-B})\epsilon^{-C} \leqslant s,\quad s \leqslant 2^{k}\epsilon,\quad s'\leqslant p(s,\epsilon).
\end{align*}
By definition $p(\cdot)$ is increasing in $s$. Therefore we can assume that
\begin{align}\label{eq:optimization_1}
 \frac{s}{\epsilon} = 2^{k}.
\end{align}
The constraint $\epsilon'\geqslant 0$ simply reduces to $\epsilon \geqslant 0$. Thus, the all constraints where $\epsilon$ is involded are
\begin{align*}
 0 \leqslant \epsilon,\quad  s \leqslant 2^{k}\epsilon,\quad  q(s,\epsilon) \leqslant \epsilon'
\end{align*}
Since $q(\cdot)$ is increasing, we can assume that $\epsilon' = q(s,\epsilon)$, or in other words that
\begin{align}\label{eq:optimization_2}
\epsilon' = a\epsilon^A.
\end{align}
Given \Cref{eq:optimization_1,eq:optimization_2} the maximum part of the optimization is eliminated. Our task reduces to minimizing the following expression
\begin{align*}
 \frac{s'}{\epsilon'} & = s\cdot \frac{c}{a}\epsilon^{C-A} - \frac{b}{a}\epsilon^{-B-A} \\
& =
  \epsilon^{-A}\left(\frac{c}{a}\cdot 2^k\epsilon^{C+1} - \frac{b}{a}\epsilon^{-B}\right).
\end{align*}
 over $(s',\epsilon')$ or equivalenty over $s,\epsilon$ (given the equalities \eqref{eq:optimization_1} and \eqref{eq:optimization_2}), provided that \Cref{eq:optimization_constraint_1} is satisfied. Thus we obtain the following problem in one variable
\begin{align}\label{eq:optimization_3}
\begin{aligned}
 \underset{\epsilon}{\mathrm{minimize}} && & a^{-1}\epsilon^{-A}\left( 2^k c\cdot \epsilon^{C+1} - b\epsilon^{-B}\right) & \\
\textrm{s.t.} 
&& &   2^k c\cdot \epsilon^{C+1} - b\epsilon^{-B} \geqslant 1.
\end{aligned}
\end{align}
Now everything depends on the behavior of the objective function
\begin{align*}
 f(u) = \frac{2^k c}{a}\cdot  u^{C+1-A} - \frac{b}{a} \cdot u^{-B-A}
\end{align*}
However the condition \eqref{eq:small_exponent} implies that $f(u)$ is increasing. Thus, it attains its minimum on the boundary point, which is given by
\begin{align}\label{eq:optimization_4}
 2^k c\cdot \epsilon^{C+1} - b\epsilon^{-B} =1 .
\end{align}
The objective function evaluated at this point gives us
\begin{align}\label{eq:optimization_5}
f(\epsilon) = a^{-1}\epsilon^{-A}
\end{align}
Note that from \Cref{eq:optimization_4} it follows that
\begin{align*}
 2^kc\cdot \epsilon^{B+C+1} = b+\epsilon^{B}
\end{align*}
If $b\geqslant 1$ we obtain $\epsilon^{B+C+1} \approx 2^{-k} b c^{-1}$ (up to a multiplicative factor of at most $2$). If $b=0$ then $\epsilon^{C+1} \approx 2^{-k}c^{-1}$.
\end{proof}

\section{Proof of \Cref{thm:simulator_new}}\label{sec:simulator_new}

\begin{proof}[of \Cref{thm:1}]
In the first step we show how to construct a simulator $h=h^{\D}$ for one circuit $\D$ of size $s$.
\begin{claim}[A perfect simulator for any any fixed real-valued distinguisher]\label{claim:1}
For any $[0,1]$-valued $\D$ of size $s$ there exists a function $h_{\D}:\mathcal{X}\rightarrow \{0,1\}^m$ of complexity $\bigO{ s\cdot2^{m} }$ such that $ \E \D(X,Z) = \E \D(X,h(X))$.
\end{claim}
\begin{proof}[of \Cref{claim:1}]
Let $\mathsf{h}^{+}_{\D}$ and $\mathsf{h}^{-}_{\D}$ be functions such that
\begin{align*}
 \D( x, \mathsf{h}^{-}_{\D}(x)) = \min_{z} \D(x,z), \quad
 \D(x, \mathsf{h}^{+}_{\D}(x)) = \max_{z} \D(x,z)
\end{align*}
Both functions can be computed by enumerating over all $z\in\{0,1\}^m$, using $2\cdot 2^m$ calls to $\D$. For any $X,Z$ we have
\begin{align*}
\E_{x\leftarrow X} \D(x, h^{-}_{\D}(x)) \leqslant \E\D(X,Z) \leqslant \E_{x\leftarrow X} \D(x, h^{+}_{\D}(x) )
\end{align*}
Therefore there exists a number $\gamma_{\D} \in [0,1]$ such that
\begin{align*}
 \E\D(X,Z) = \gamma_{\D} \E_{x\leftarrow X} \D(x, h^{-}_{\D}(x)) + (1-\gamma_{\D})\E_{x\leftarrow X} \D(x,h^{+}_{\D}(x)). 
\end{align*}
We define $h(x)=h_{\D}(x)$ as follows: sample $r\in [0,1]$; if $r\leqslant \gamma$ then we output $h^{-}_{\D}(x)$ else we output $h^{+}_{\D}(x)$.
\end{proof}
Now we apply the min-max theorem in a standard  way to change the order of quantifiers. 
\begin{claim}[One simulator for all distinguishers]\label{claim:3}
There exists a distribution $\overline{\mathsf{h}}$ on functions $h$ of complexity $\bigO{ s\cdot 2^m }$ such that $\left| \E\D(X,Z) - \E_{h\leftarrow \bar{h}}\D(X,h(X))\right|\leqslant \epsilon$ for all $\D$ of size $s\epsilon^2$.
\end{claim}
\begin{proof}
By a standard application of the min-max theorem combined with the Chernoff Bound (see~\cite{Barak03} for esentially the same technique) we get that there is a distribution $\overline{h}$ such that for all $\D$ of size $s\epsilon^2$ we have $\E\D(X,Z) - \E_{h\leftarrow \bar{h}}\D(X,h(X)) \leqslant \epsilon$. Since this holds for $\D$ and $\D^c$ for any $\D$ of size $s$, the result follows.
\end{proof}
In the last step we approximate this possibly inefficient simulator in the statistical distance.
\begin{claim}[One efficient simulator for all distinguishers]\label{claim:4}
There exists a simulator $\mathsf{h}$ of complexity $\bigO{ s\cdot 2^{2m}\epsilon^{-2} }$ such that $|\E\D(X,Z) - \E\D(X,h(X))| \leqslant 2\epsilon$ for all $\D$ of size $s\epsilon^2$.
\end{claim}
\begin{proof}[Proof of \Cref{claim:4}]
Let $h_0$ be the inefficient simulator guaranteed by \Cref{claim:3}.
We know that $h_0$ is of the following form
\begin{align}
 \mathbf{P}_{X,h_0(X)} = \E_{h\leftarrow \bar{h}} \mathbf{P}_{X,h(X)} =  \E_{x\leftarrow X}\E_{h\leftarrow \bar{h}}\mathbf{P}_{x,h(x)}
\end{align}
Fix a number $t$ and sample $h_j \leftarrow \bar{h}$ for $j=1,\ldots,t$. For a fixed choice of  $h_1,\ldots,h_t$ we define the randomized function $\tilde{\mathsf{h}}(x)$ as follows: $\mathbf{P}_{\tilde{\mathsf{h}}(x)}(z) = t^{-1}\sum_{i=1}^{t}\mathbf{P}_{\mathsf{h}_j(x)}(z)$ (it simply takes $i\leftarrow \{1,\ldots,t\}$ and outputs $h_i$). Below we assume that $x$ is sampled according to $X$. 
Let us compute
\begin{align*}
\E_{\{h_j\}_{j=1}^{t}}\E_{x}\left\|\mathbf{P}_{\tilde{h}(x)}(\cdot)- \mathbf{P}_{\bar{h}(x)}(\cdot)\right\|_2^2& =  t^{-2} \E_{x}\E_{\{h_i\}_{i=1}^{t}}\left\|\sum_{j=1}^{t}\left(\mathbf{P}_{h_j(x)}(\cdot)- \mathbf{P}_{\bar{h}(x)}(\cdot)\right)\right\|^{2}_2 \nonumber \\
 = & t^{-2} \E_{x }\left[\sum_{j=1}^{t} \E_{h_j}\left\|\mathbf{P}_{h_j(x)}(\cdot)- \mathbf{P}_{\bar{h}(x)}\right\|^{2}_2\right] \nonumber  \\
= & t^{-1} \left( \E_{h\leftarrow\bar{h}} \E_{x }\left\| \mathbf{P}_{h(x)}(\cdot) \right\|^2_2 -  \E_{x}\left\| \mathbf{P}_{\bar{h}(x)}(\cdot) \right\|^2_2\right) 
\end{align*}
Therefore for some choice of $h_1,\ldots,h_t$ we have
\begin{align}
\E_{x}\left\|\mathbf{P}_{\tilde{h}(x)}(\cdot)- \mathbf{P}_{\bar{h}(x)}(\cdot)\right\|_2^2 \leqslant 
\frac{1}{t}
\label{eq:approximation_second_moment}
\end{align}
Note that the simple probabilistic proof of \Cref{eq:approximation_second_moment} resembles the proof of Maurey-Jones-Barron Theorem (see Lemma 1 in ~\cite{Barron93universalapproximation}) on approximating convex hulls in Hilbert spaces. Using the fact that $|\D(\cdot,\cdot) | \leqslant 1$ and inequality between the first and the second norm
\begin{align}
 \left|\E\D(X,\tilde{h}(X))-\E\D(X,\bar{h}(X))\right| & = \E_{x }\left| \E\D(x,h(x))-\E\D(x,\bar{h}(x)) \right| \nonumber \\
 & \leqslant \E_{x}\left\|  \mathbf{P}_{\tilde{h}(x)}(\cdot)-\mathbf{P}_{\bar{h}(x)} (\cdot)  \right\|_1 \nonumber \\
 & \leqslant 2^{m/2}\cdot \left(\E_{x }\left\| \mathbf{P}_{\tilde{h}(x)}(\cdot)-\mathbf{P}_{\bar{h}(x)}(\cdot) \right\|_2\right)^{\frac{1}{2}}
 \label{eq:approximation_advantage}
\end{align}
Combining \Cref{eq:approximation_second_moment} and \Cref{eq:approximation_advantage} we get for some choices of $h_1,\ldots,h_{t}$ 
\begin{align}
 \left|\E\D(X,\tilde{h}(X))-\E\D(X,\bar{h}(X))\right| \leqslant \left(2^{m}t^{-1}\right)^{\frac{1}{2}}
\end{align}
Setting $t = 2^{m}\epsilon^{-2}$ we finish the proof.
\end{proof}
The result follows now directly from \Cref{claim:4}, for \emph{real-valued} circuits. Up to an error of $\delta=2^{-\rho}$ in the advantage, we can approximate them by circuits taking values in the discrete set $\left\{2^{-\rho},2\cdot 2^{-\rho},\ldots,1\right\}$.
Any such a circuit $\D$ we start our proof with, can be viewed as a combination of 
\begin{enumerate}[(a)]
\item the coding vector $\left(\D^{(i)}\right)_{i=1}^{\rho}$ of $\rho$ circuits of size $s$, computing the first $\rho$ digits of the binary expansion of the output
\item the decoding circuit of size $2\rho$ which uses
 additional $\rho$ random bits to read $\left( \D^{(i)}(x,z)\right)_{i} $ and to output $1$ with probability $\D(x,z)$ (in the $i$-th round it toss a coin and either halts and outputs $\D^{i}(x,z)$ or it moves to the round $i+1$; in round $n+1$ the output is $0$) 
\end{enumerate}
Now, the correct complexity for $h$ in \Cref{claim:1,claim:2} is $\bigO{s\cdot 2^{m} \rho }$, by the use of sorting network. Setting $\rho = \log(1/\epsilon)$ we see that we lose $\bigO{s\cdot 2^{m}\log(\epsilon^{-1})}$ in the simulator complexity.
\end{proof}

\end{document}